\documentclass[]{interact}

\usepackage{epstopdf}% To incorporate .eps illustrations using PDFLaTeX, etc.
\usepackage{subfigure}% Support for small, `sub' figures and tables
\usepackage{multirow}
\usepackage{float}
\usepackage[numbers,sort&compress]{natbib}% Citation support using natbib.sty
\bibpunct[, ]{[}{]}{,}{n}{,}{,}% Citation support using natbib.sty
\usepackage{tabularx}
\usepackage{booktabs}

\usepackage{lineno,hyperref}

\theoremstyle{plain}% Theorem-like structures
\newtheorem{theorem}{Theorem}[section]
\newtheorem{lemma}[theorem]{Lemma}
\newtheorem{corollary}[theorem]{Corollary}
\newtheorem{proposition}[theorem]{Proposition}

\theoremstyle{definition}

\theoremstyle{remark}
\newtheorem{remark}{Remark}

\begin{document}

\articletype{ \ \ }

\title{Objective Bayesian inference for the Dhillon distribution}

\author{ Pedro Luiz Ramos$^{\rm 1}$$^{\ast}$\thanks{$^\ast$Corresponding author. Email: pedro.ramos@uc.cl
\vspace{6pt}}, Enrique Achire Quispe$^{\rm 1}$, Ricardo Puziol de Oliveira$^{\rm 2}$ \\ and Jorge A. Achcar$^{\rm 3}$ \\\vspace{6pt} 
\normalsize{$^{1}$Facultad de Matemáticas, Pontificia Universidad Católica de Chile, \\ Santiago, Chile} \\
\normalsize{$^{2}$Department of Statistics, State University of Sao Paulo, Presidente \\ Prudente, Brazil} \\
\normalsize{$^{3}$Ribeirao Preto Medical School, University of Sao Paulo, Ribeirao \\ Preto, Sao Paulo, Brazil} \\
}

\maketitle

\begin{abstract}
In this work, we develop an objective Bayesian framework for the Dhillon probability distribution. We explicitly derive three objective priors: the Jeffreys prior, the overall reference prior, and the maximal data information prior. We show that both Jeffreys and reference priors yields a proper posterior distribution, whereas the maximal data information prior leads to an improper posterior. Moreover, we establish sufficient conditions for the existence of its respective posterior moments. Bayesian inference is carried out via Markov chain Monte Carlo, using the Metropolis–Hastings algorithm. A comprehensive simulation study compares the Bayesian estimators to maximum likelihood estimators in terms of bias, mean squared error, and coverage probability. Finally, a real-data application illustrates the practical utility of the proposed Bayesian approach.
\end{abstract}


\section{Introduction}

Ensuring the reliability of electronic systems, such as semiconductor devices, power modules, and control circuits, is essential to prevent unexpected failures, extend service life, and reduce maintenance costs. Classical lifetime laws like the exponential, Weibull, and Gamma models impose strictly monotonic hazard rates, which are unable to capture both an initial burn‐in period and a subsequent wear‐out phase. To address this limitation, Dhillon \cite{dhillon1980statistical} introduced a two‐parameter hazard‐rate model whose instantaneous failure rate is given by
\[
h(t)=\frac{\beta\,\theta\,t^{\beta-1}}{1+\theta\,t^\beta}.
\]
 where it is strictly decreasing for \(0<\beta\le 1\) and unimodal (increases to a single peak and then decreases) for \(\beta>1\), with a unique maximum. Hence, the model accommodates components with declining early-life risk as well as systems whose risk rises to a single peak before decreasing---i.e., unimodal wear-out behavior. Dhillon’s motivation was to better describe infant‐mortality failures and aging effects in microelectronic components, and this flexible yet tractable hazard formulation has since become a cornerstone of modern reliability engineering. Some extensions for the Dhillon distribution have been considered in recent years. Abba et al. \cite{abba2025bi} introduced a four-parameter Bi-Failure Modes model in which the distribution of failure times due to one failure type follows a Dhillon law, while that due to another cause follows an exponential-power model. Abba et al. \cite{abba2025bayesian} used the same approach to construct a Weibull–Dhillon competing risk model, with further robustifications in \cite{abba2024robust}.

An important drawback of Dhillon’s original presentation is that he neither derives the full log‐likelihood function nor computes the score equations and Fisher information matrix for the parameters vector, which is essential for constructing asymptotic confidence intervals. Consequently, maximum‐likelihood estimators remain unspecified in specified form, and their asymptotic variances cannot be obtained via the standard expected expressions. Moreover, the model has not received a dedicated Bayesian treatment: no posterior distribution has been formulated, nor has the question of posterior propriety been addressed. This gap leaves both frequentist and Bayesian properties of the estimators unexplored.

Under the Bayesian approach, objective priors are often chosen when informative priors are unavailable, so the data drive the posterior distribution, but their impropriety can prevent valid inference by producing improper posteriors. Common noninformative priors include Jeffreys’ prior and reference priors \cite{jeffreys1946invariant, berger1992reference}. Northrop and Attalides \cite{northrop2016} point out that simple rules for ensuring posterior propriety are lacking in this setting and providing such conditions are of main concern \cite{ramos2023power}. Ramos et al. \cite{ramos2020posterior} applied similar criteria to survival models, illustrating posterior propriety in complex settings. These results are extended to provide explicit conditions for the existence of posterior moments, including means and variances, ensuring their mathematical validity \cite{ ramos2019improved, ramos2020posterior, ramos2021posterior, ramos2024objective, ramos2025posterior, achire2025posterior}.

We consider several widely used objective priors, such as the Jeffreys' prior \cite{jeffreys1946invariant}, the reference prior of Berger and Bernardo \cite{berger1992reference, bernardo2005}, and Zellner’s maximal data information prior (MDIP) \cite{zellner1, zellner2}. It is proven that each prior yields a proper posterior with finite moments, guaranteeing valid posterior inference. Bayesian computations are carried out via Markov chain Monte Carlo (MCMC) using the Metropolis--Hastings algorithm to sample from the posterior and compute relevant marginal statistics. Comparative analyses against maximum likelihood estimators demonstrate improvements in bias, mean squared error, and nominal coverage. To showcase the Dhillon distribution’s practical value, we analyze two real‐world failure‐time datasets drawn from sugarcane‐harvester machinery operating under severe loads and highly abrasive conditions. The first dataset records 66 times (in days) between corrective interventions of a diesel engine, and the second comprises 90 time‐to‐failure observations for a line‐divider mechanism. These components exhibit complex wear patterns that evolve with usage, making accurate lifetime modeling essential for forecasting the next breakdown.

Overall, despite the recognized flexibility of the Dhillon distribution in modeling non-monotonic hazard rates, 
there is no published work providing a complete objective Bayesian treatment for this model, including proofs of posterior propriety, 
closed-form derivation of reference priors, and a systematic comparison with classical estimators in simulated and real-world settings. 
This article addresses this gap by developing and validating such a framework, combining theoretical results, 
extensive simulation studies, and applied analysis of maintenance data from agricultural machinery.

The remainder of the paper is organized as follows. Section~2 reviews the Dhillon distribution, its survival and hazard properties, and related functions. Section~3 introduces the maximum-likelihood estimation framework. Section~4 develops the objective Bayesian approach: after presenting the general posterior setup, we derive and analyze Jeffreys, reference and MDIP, and describe the Metropolis–Hastings sampler used for posterior computation. Section~5 reports a comprehensive simulation study; Section~6 describes the posterior predictive distribution for future lifetimes. Section~7 illustrates the methodology on real failure-time data from harvester components, and Section~8 concludes with summary and future directions.

\section{Background and properties}

Let \(T\) denote the lifetime of a component following a Dhillon distribution with shape parameter \(\beta>0\) and scale parameter \(\theta>0\), denoted as \(\mathrm{Dhillon}(\beta,\theta)\). The survival  function is given by:
\begin{equation}\label{survival}
    R(t) = \frac{1}{1 + \theta\,t^\beta}.
\end{equation}

Its probability density function obtained from the formula $f(t)=h(t)R(t)$, is given by:
\begin{equation}\label{eq-0}
f(t;\beta,\theta)
=\frac{\beta\,\theta\,t^{\beta-1}}{\bigl(1+\theta\,t^\beta\bigr)^{2}}, 
\quad t>0.
\end{equation}

The \(r\)-th moment exists if and only if \(0 < r < \beta\) and can be expressed in closed form (see Appendix~\ref{apendix-A}) as
\[
E\bigl[T^r\bigr]
=\frac{r\pi}{\beta\,\sin\!\bigl(\pi\,r/\beta\bigr)}\;\theta^{-\,r/\beta}\,.
\]

In particular, the first two moments—mean and variance—follow as
\[
E[T]
=\frac{\pi}{\beta\,\sin\!\bigl(\pi/\beta\bigr)}\;\theta^{-\,1/\beta},
\qquad
\mathrm{Var}(T)
=\theta^{-\,2/\beta}\biggl\{\frac{2\pi}{\beta\,\sin(2\pi/\beta)}
-\frac{\pi^2}{\beta^2\,\sin^2(\pi/\beta)}\biggr\}.
\]

The behavior of the hazard function \(h(t)\) as \(t\to0\) and \(t\to\infty\) is
\[
h(0)=
\begin{cases}
\infty, & 0<\beta<1,\\
0,      & \beta>1,
\end{cases}
\qquad
h(\infty)=0,
\quad \beta>0.
\]

\begin{theorem}\label{glaser_dhillon}
For \(\beta>0\) and \(\theta>0\), the hazard function of the \(\mathrm{Dhillon}(\beta,\theta)\) distribution satisfies:
\[
h(t)\;
\begin{cases}
\text{is strictly decreasing on }(0,\infty), 
&0<\beta\le1,\\
\text{is unimodal with a unique maximum at }
t^*=\bigl(\tfrac{\beta-1}{\theta}\bigr)^{1/\beta},
&\beta>1.
\end{cases}
\]
\end{theorem}

\begin{proof}
From the Dhillon model, one finds
\[
-\frac{h'(t)}{(h(t))^2}
=\frac{\theta\,t^\beta - (\beta-1)}{\beta\,\theta\,t^\beta}.
\]
First, assuming that \(0<\beta\le1\).  Then \(\beta-1\le0\), so \(\theta\,t^\beta - (\beta-1)>0\) for all \(t>0\), hence \(h'(t)<0\) and \(h(t)\) is strictly decreasing.

In the case where \(\beta>1\), then \(\theta\,t^\beta - (\beta-1)\) changes sign once at 
\[
t^*=\bigl(\tfrac{\beta-1}{\theta}\bigr)^{1/\beta},
\]
so \(h'(t)>0\) on \((0,t^*)\) and \(h'(t)<0\) on \((t^*,\infty)\).  Thus, \(h(t)\) is unimodal with its unique maximum at \(t^*\).
\end{proof}

Figure \ref{denshazard} illustrates both the strictly decreasing hazard (when \(0<\beta\le1\)) and the unimodal hazard (when \(\beta>1\)) of the Dhillon distribution for various \(\beta\) and \(\theta\).

\begin{figure}[!h]
\centering
\includegraphics[scale=0.7]{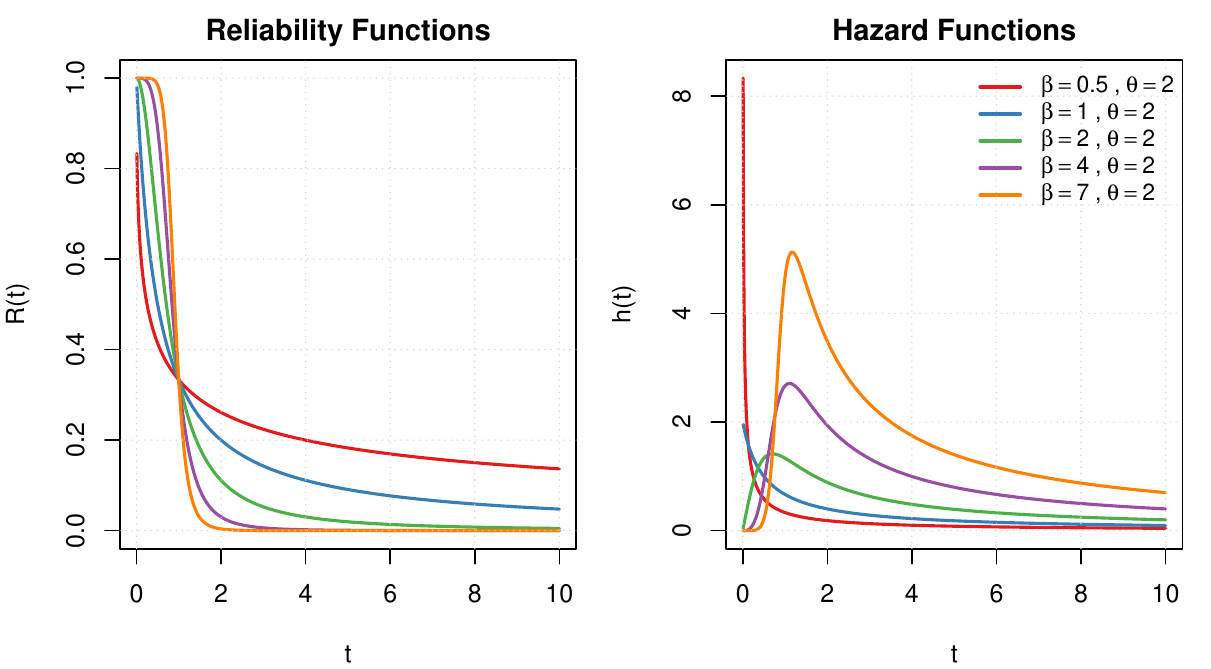}
\caption{Survival and hazard functions of the Dhillon distribution for different \(\beta\).}\label{denshazard}
\end{figure}

Examining Figure \ref{denshazard}, the Dhillon distribution’s hazard rate rises from zero, attains its maximum at \(t^*\), and then falls back toward zero as \(t\to\infty\).  Likewise, the PDF displays a sharp mode at \(t^*\) followed by a heavy right‐hand tail.  This combination of an early surge in failures and occasional late‐time events makes the Dhillon model especially suitable for systems—such as high‐stress repaired devices—that suffer a burst of failures shortly after restart but still exhibit sporadic failures thereafter.  We illustrate this behavior with three repaired‐equipment datasets.  

Another important characteristic is the mean residual life (MRL) function, which measures the expected additional lifetime given survival up to time \(t\).

\begin{proposition}
The mean residual life function of the $\mathrm{Dhillon}(\beta,\theta)$ distribution exists if and only if $\beta > 1$, and is given by  
\[
r(t)
= \frac{\theta^{-1/\beta} \,\bigl(1+\theta\,t^\beta\bigr)}{\beta}
\,B\left(\frac{1}{1+\theta t^\beta};\, 1-\frac{1}{\beta},\,\frac{1}{\beta}\right),
\]
where $B(z; a, b) = \int_{0}^{z} u^{\,a-1} (1-u)^{\,b-1} \, du$ denotes the incomplete beta function.  

The corresponding asymptotic limits are  
\[
r(0)
= \theta^{-1/\beta} \,\frac{\pi}{\beta\,\sin\!\bigl(\pi/\beta\bigr)},
\qquad
r(\infty) = \infty.
\]
\end{proposition}
\begin{proof}
    The proof is provided in Appendix~\ref{apendix-A}.
\end{proof}
The following lemma, adapted from \cite{tang2002mean}, will play a key role in the direct proof of Theorem~\ref{theorem-mlr}.

\begin{lemma}\label{brison_dhillon}  
If a lifetime \(T\) has hazard \(h(t)\) unimodal and \(h(0)\,r(0)<1\), then its MRL \(r(t)\) has a bathtub shape (strictly decreasing up to some \(t_0\), then strictly increasing).
\end{lemma}

\begin{theorem}\label{theorem-mlr}
Let \(\theta>0\) and \(\beta>1\).  Then the mean residual life function \(r(t)\) of \(\mathrm{Dhillon}(\beta,\theta)\)  has a bathtub shape (decreases then increases). 
\end{theorem}

\begin{proof}
For \(\beta>1\), we have \(h(t)\) unimodal (Theorem \ref{glaser_dhillon}) and one checks \(h(0)\,r(0)<1\).  By Lemma \ref{brison_dhillon}, \(r(t)\) first decreases to a unique minimum \(t_0\) and then increases to \(+\infty\).
\end{proof}

Figure \ref{dhillon_mrl} illustrates  the bathtub‐shaped MRL for \(\beta>1\).

\begin{figure}[!h]
\centering
\includegraphics[scale=0.7]{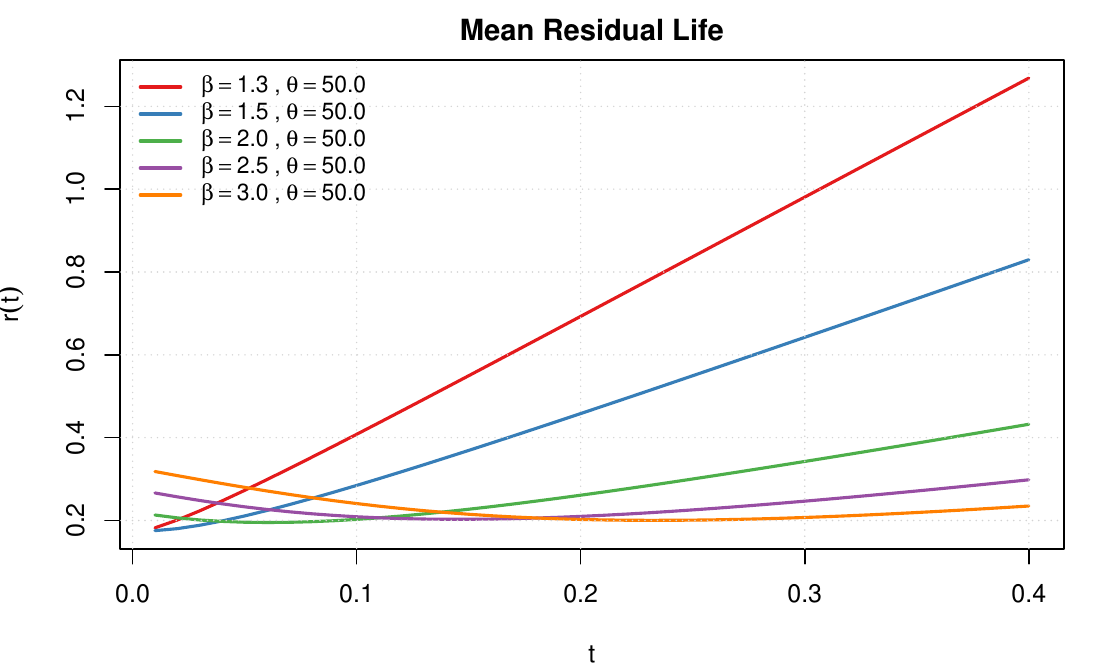}
\caption{Mean residual functions of the Dhillon distribution for different \(\beta\).}\label{dhillon_mrl}
\end{figure}

\section{Maximum likelihood estimator}

Assume now that we observe a complete random sample \(T_1,\dots,T_n\) drawn from the \(\mathrm{Dhillon}(\beta,\theta)\) distribution.  The likelihood function for the parameters \(\beta\) and \(\theta\), given the observed data \(\boldsymbol{t}=(t_1,\dots,t_n)\), is
\begin{equation}\label{eq-8}
L(\beta,\theta\mid\boldsymbol{t})
=\beta^n\,\theta^n\prod_{i=1}^n
\,t_i^{\beta-1}\bigl(1+\theta\,t_i^\beta\bigr)^{-2}.
\end{equation}
The corresponding log-likelihood is
\[
\mathcal{L}(\beta,\theta\mid\boldsymbol{t})
=n\ln\beta + n\ln\theta + (\beta-1)\sum_{i=1}^n\ln t_i
-2\sum_{i=1}^n\ln\bigl(1+\theta\,t_i^\beta\bigr).
\]

The maximum‐likelihood estimators \(\widehat\beta\) and \(\widehat\theta\) are found by solving the score equations obtained from the first derivatives of \(\mathcal{L}\):
\begin{align}
\frac{\partial\mathcal{L}}{\partial\beta}
&= \frac{n}{\beta}
+ \sum_{i=1}^n\ln t_i
-2\sum_{i=1}^n
\frac{\theta\,t_i^\beta\,\ln t_i}{1+\theta\,t_i^\beta}
\;=\;0,
\label{eq-dh-score1}\\
\frac{\partial\mathcal{L}}{\partial\theta}
&= \frac{n}{\theta}
-2\sum_{i=1}^n
\frac{t_i^\beta}{1+\theta\,t_i^\beta}
\;=\;0.
\label{eq-dh-score2}
\end{align}
These equations do not admit closed‐form solutions, so one must employ numerical optimization routines—such as the Newton–Raphson algorithm or quasi‐Newton methods—to obtain \(\widehat\beta\) and \(\widehat\theta\). Under standard regularity conditions, the expected Fisher information matrix for \((\beta,\theta)\) can be obtained in closed form. The details of its derivation are provided in Appendix~\ref{apendix-A}. We have
\begin{equation}\label{eq-dh-fisher}
I(\beta,\theta)
= n
\begin{pmatrix}
\dfrac{\pi^2+3+3\log^2(\theta)}{9\,\beta^2} & -\dfrac{\log \theta}{3\theta\beta} \\[8pt]
-\dfrac{\log \theta}{3\theta\beta} & \dfrac{1}{3\,\theta^2}
\end{pmatrix}.
\end{equation}

Hence the marginal asymptotic distributions are
\[
\widehat\beta\;\overset{\mathrm{approx}}{\sim}\;
N\!\left(\beta,\;
\frac{9\,\beta^2}{n(\pi^2+3)}\right),
\qquad
\widehat\theta\;\overset{\mathrm{approx}}{\sim}\;
N\!\left(\theta,\;
\frac{3\,\theta^2(\pi^2+3+3\log^2\theta)}{n(\pi^2+3)}\right).
\]
Let \(z_{1-\alpha/2}\) be the \((1-\alpha/2)\)-quantile of the standard normal. Then approximate \(100(1-\alpha)\%\) confidence intervals for \(\beta\) and \(\theta\) are
\[
\widehat\beta
\pm
z_{1-\alpha/2}\,
\sqrt{\frac{9\,\widehat\beta^2}{n(\pi^2+3)}}
\quad\text{and}\quad
\widehat\theta
\pm
z_{1-\alpha/2}\,
\sqrt{\frac{3\,\widehat\theta^2(\pi^2+3+3\log^2\widehat\theta)}{n(\pi^2+3)}},
\]
respectively, 
where in practice the unknown true parameters inside the variances are replaced by their MLEs. These intervals rely on the large-sample normal approximation and ignore the correlation between \(\widehat\beta\) and \(\widehat\theta\) when taken marginally.
\section{Bayesian Inference}\label{section-4}

In the Bayesian framework, inference is based on the posterior distribution, which combines the likelihood—containing the information provided by the data—with the prior distribution—reflecting any pre-existing knowledge about the parameters. For the Dhillon distribution with parameters \((\beta,\theta)\), the joint posterior density is given by
\begin{equation}\label{posterior-dhillon}
p(\beta,\theta \mid \boldsymbol{t})
\;=\;
\frac{\pi(\beta,\theta)\,\beta^n\,\theta^n}
     {d(\boldsymbol{t})}
\;\prod_{i=1}^n
t_i^{\beta-1}\,\bigl(1+\theta\,t_i^\beta\bigr)^{-2},
\end{equation}
where the normalizing constant is
\begin{equation}\label{normalizing-constant}
d(\boldsymbol{t})
=\int_{0}^{\infty}\!\!\int_{0}^{\infty}
\pi(\beta,\theta)\,\beta^n\,\theta^n
\;\prod_{i=1}^n
t_i^{\beta-1}\,\bigl(1+\theta\,t_i^\beta\bigr)^{-2}
\,d\beta\,d\theta.
\end{equation}

When no prior information about the parameters is available, it is common to adopt an objective prior for \(\pi(\beta,\theta)\). However, since objective priors are often improper, it is essential to verify that \(d(\boldsymbol{t}) < \infty\) to ensure the posterior is proper.

Following Ramos et al.\ \cite{ramos2023power}, the convergence of \eqref{normalizing-constant} depends on the tail behavior of \(\pi(\beta,\theta)\). Heavy-tailed priors in either \(\beta\) or \(\theta\) may lead to divergence, whereas sufficiently light-tailed priors guarantee propriety.

In what follows, we derive several objective priors for the parameters \(\beta\) and \(\theta\) of the Dhillon distribution using formal criteria, such as Jeffreys' prior, the MDIP, and reference priors. We also assess whether the resulting posterior distributions are proper under each approach.

\subsection{Jeffreys prior}

The Jeffreys prior \cite{jeffreys1946invariant} is a cornerstone of objective Bayesian analysis on account of its invariance under smooth reparametrizations and its minimal informational influence.  By construction it assigns to each parameter value a weight proportional to the square root of the local information content, ensuring that the prior measure transforms correctly if one changes variables.  Concretely, for a parameter vector \((\beta,\theta)\) the Jeffreys prior is defined as
\[
\pi_{J}(\beta,\theta)
\;\propto\;
\sqrt{\det I(\beta,\theta)},
\]
where \(I(\beta,\theta)\) is the Fisher information matrix per observation.

For the Dhillon distribution the  Jeffreys prior for \(\beta\) and \(\theta\) reduces to
\begin{equation}\label{eq-jeffreys-dhillon}
\pi_{J}(\beta,\theta)
\;\propto\;
\frac{1}{\beta\,\theta}.
\end{equation}

The joint posterior distribution for \(\beta\) and \(\theta\), obtained using the Jeffreys prior \eqref{eq-jeffreys-dhillon}, is proportional to the product of the likelihood \eqref{eq-8} and the prior, yielding
\begin{equation}\label{postj-dhillon}
\begin{aligned}
p_{J}(\beta,\theta\mid\boldsymbol{t})
&\propto \beta^{n-1}\,\theta^{\,n-1}\;\prod_{i=1}^{n}
t_{i}^{\beta}\,\bigl(1+\theta\,t_{i}^{\beta}\bigr)^{-2}\,. 
\end{aligned}
\end{equation}

\begin{proposition}\label{prop2} The posterior (\ref{postj-dhillon}) is proper for $n\geq 2$ as long as not all $t_i$ are equal.
\end{proposition}

\begin{proof}
Let $r = t_k/t_j < 1$, where $t_k = \min_i t_i$ and $t_j = \max_i t_i$. By applying Fubini-Tonelli's Theorem and noting that $\theta t_i^{\beta}(1+\theta t_i^\beta)^{-2} < 1$, we obtain
\begin{equation}\label{eq-1}
\begin{aligned}
        d(\boldsymbol{t}) 
        &= \int_{0}^{\infty}\!\!\int_{0}^{\infty}
        \beta^{n-1}\,\theta^{-1}
        \prod_{i=1}^n
        \theta t_i^{\beta}\bigl(1+\theta\,t_i^\beta\bigr)^{-2}
        \;d\theta\,d\beta \\
        &\leq \int_{0}^{\infty}\!\!\int_{0}^{\infty}
        \beta^{n-1}\,\theta^{-1}
        \left[\theta t_k^{\beta}\bigl(1+\theta\,t_k^\beta\bigr)^{-2}
        \cdot \theta t_j^{\beta}\bigl(1+\theta\,t_j^\beta\bigr)^{-2}\right]
        \;d\theta\,d\beta \\
        &\leq \int_{0}^{\infty}
        \beta^{n-1} 
        \int_{0}^{\infty} \theta t_k^{\beta}  t_j^{\beta} 
        \bigl(1+\theta\,t_k^\beta\bigr)^{-2}
        \bigl(1+\theta\,t_j^\beta\bigr)^{-2}
        \;d\theta\,d\beta.
    \end{aligned}
\end{equation}

Now, perform the change of variables $x = \theta t_j^\beta$. From (\ref{eq-1}), we obtain
\begin{equation*}
\begin{aligned}
d(\boldsymbol{t}) 
\leq \int_{0}^{\infty}
\beta^{n-1} r^\beta 
\int_{0}^{\infty} \frac{x}
{(1+r^\beta x)^2(1+x)^2} 
\;dx \,d\beta 
= \int_{0}^{\infty}
\beta^{n-1} r^\beta 
J(\beta)
\,d\beta,
\end{aligned}
\end{equation*}
where
\begin{equation}\label{J-integral}
    \begin{aligned}
      J(\beta) = \int_{0}^{\infty} \frac{x}
{(1+r^\beta x)^2(1+x)^2} 
\;dx 
= \frac{\beta(r^\beta+1)\log r - 2(r^\beta - 1)}{(r^\beta - 1)^3}.
    \end{aligned}
\end{equation}

As $\beta \to 0$, we have $r^\beta \to 1$ and $J(\beta) \to 1/6$, hence $\beta^{n-1}r^\beta J(\beta) \propto \beta^{n-1}$, which is integrable near $0$ for $n \geq 2$.

As $\beta \to \infty$, we have $J(\beta) \propto \beta$, so $\beta^{n-1}r^\beta J(\beta) \propto \beta^n r^\beta$, which is integrable near $\infty$ since $r < 1$.

Therefore, we conclude that
\[
\int_{0}^{\infty}
\beta^{n-1} r^\beta 
J(\beta)
\,d\beta < \infty,
\]
which shows that \( d(\boldsymbol{t}) < \infty \).
\end{proof}
As a direct consequence of the proof of the previous proposition, we conclude that all posterior moments for the parameter \(\beta\) are finite (see Proposition~\ref{prop-moments-beta} in the appendix); however, the posterior moments for the parameter \(\theta\) do not always exist.  
In Remark~\ref{rmk-1} in the appendix, we observed that for \(n=2\) and observations satisfying \(t_1 < t_2 < 1\), the posterior mean of \(\theta\) fails to exist.  
The following corollary, which follows from Proposition~\ref{prop-moments-theta}, establishes necessary (although not sufficient) conditions for the existence of the posterior mean of \(\theta\).  
These conditions pertain to the sample and require the presence of observations both below and above 1, a requirement that is typically satisfied for generic samples of moderate size.

\begin{corollary}
The posterior means of \(\beta\) and \(\theta\), under the posterior distribution (\ref{postj-dhillon}), are finite for \(n > 2\), provided that there is at least one observation less than 1 and at least two observations greater than 1.
\end{corollary}

\subsection{Reference prior}

Another important class of noninformative priors is the reference prior, originally introduced by Berger and Bernardo \cite{berger1992development}. In this approach one first decides which parameter is of primary interest and which are treated as nuisance. The objective is to select a prior \(\pi(\Theta)\) that maximizes the expected information about the parameter(s) of interest provided by the data, relative to the information in the prior. A suitable measure of this expected information is
\[
I(\Theta)
=
\mathbb{E}_{\boldsymbol{x}}
\bigl[
K\bigl(\pi(\Theta\mid\boldsymbol{x}),\pi(\Theta)\bigr)
\bigr],
\]
where
\[
K\bigl(\pi(\Theta\mid\boldsymbol{x}),\pi(\Theta)\bigr)
=
\int
\pi(\Theta\mid\boldsymbol{x})
\ln\frac{\pi(\Theta\mid\boldsymbol{x})}{\pi(\Theta)}
\,d\Theta
\]
is the Kullback–Leibler divergence between posterior and prior. The reference prior is defined as the prior that maximizes \(I(\Theta)\), averaged over samples, following a prescribed ordering when nuisance parameters are present \cite{berger1992development}. Bernardo \cite{bernardo2005} showed that under asymptotic normality the ordered reference prior for the parameter of interest \(\psi\) with nuisance \(\lambda\) can be built from the Schur complements: if \(I_{\psi\psi\cdot\lambda}=I_{\psi\psi}-I_{\psi\lambda}^2/I_{\lambda\lambda}\), then the prior for \(\psi\) is proportional to \(\sqrt{I_{\psi\psi\cdot\lambda}}\sqrt{I_{\lambda\lambda}}\) .

Applying this to the Dhillon\((\beta,\theta)\) model, if \(\beta\) is the parameter of interest and \(\theta\) is nuisance, then
\[
I_{\beta\beta\cdot\theta}
= I_{\beta\beta} - \frac{I_{\beta\theta}^2}{I_{\theta\theta}}
= \frac{\pi^2+3}{9\,\beta^2},
\qquad
I_{\theta\theta}
= \frac{1}{3\,\theta^2},
\]
so that
\[
\pi_R(\beta,\theta)
\propto \sqrt{I_{\beta\beta\cdot\theta}}\;\sqrt{I_{\theta\theta}}
\propto \frac{1}{\beta\,\theta}.
\]
If instead \(\theta\) is of interest and \(\beta\) is nuisance, then
\[
I_{\theta\theta\cdot\beta}
= I_{\theta\theta} - \frac{I_{\beta\theta}^2}{I_{\beta\beta}}
= \frac{\pi^2+3}{3\,\theta^2(\pi^2+3+3\log^2\theta)},
\qquad
I_{\beta\beta}
= \frac{\pi^2+3+3\log^2\theta}{9\,\beta^2},
\]
and similarly
\[
\pi_R(\beta,\theta)
\propto \sqrt{I_{\theta\theta\cdot\beta}}\;\sqrt{I_{\beta\beta}}
\propto \frac{1}{\beta\,\theta}.
\]
Thus, under either ordering the ordered reference prior coincides:
\begin{equation}\label{priorref_updated}
\pi_R(\beta,\theta)\propto \frac{1}{\beta\,\theta}.
\end{equation}

This coincidence is not merely algebraic; it reflects the invariance structure of the Dhillon model's Fisher information, 
which results in identical functional forms for Jeffreys’ prior and the ordered reference prior, regardless of the parameter ordering. 
Consequently, subsequent Bayesian analyses under these priors yield equivalent posterior distributions.

The joint posterior distribution for \(\beta\) and \(\theta\), using \(\pi_R(\beta,\theta)\), is proportional to the product of the likelihood \eqref{eq-8} and the prior, which has the same form as presented in \eqref{postj-dhillon}.

\subsection{Maximal data information prior (MDIP)}

It is critical that the experimental data contributes additional information about the parameters beyond what is provided by the prior density; otherwise, performing the experiment would not be justified. Consequently, we are interested in a prior distribution that maximizes the incremental information extracted from the likelihood relative to the information already encoded in the prior. In line with this reasoning, Zellner \cite{zellner1} derived a family of priors by optimizing the average information present in the data density compared to the information inherent in the prior.

Zellner \cite{zellner2} employed Shannon entropy to build an information‐theoretic criterion for noninformative prior selection. Shannon defines entropy as
\begin{equation*}
H(\beta,\theta)
=\int_{0}^{\infty}f(t\mid\beta,\theta)\,\log f(t\mid\beta,\theta)\,dt,
\end{equation*}
which (up to sign) measures the information content of the density \(f(t\mid\beta,\theta)\). Accordingly, the MDIP functional is
\begin{equation*}
G[\pi(\beta,\theta)]
=\int_{0}^{\infty}H(\beta,\theta)\,\pi(\beta,\theta)\,d\beta\,d\theta
\;-\;
\int_{0}^{\infty}\pi(\beta,\theta)\,\log\pi(\beta,\theta)\,d\beta\,d\theta,
\end{equation*}
where the first term aggregates the model’s informational content under \(\pi\) and the second penalizes the prior’s own information. Imposing \(\int\pi(\beta,\theta)\,d\beta\,d\theta=1\) and maximizing \(G[\pi]\) yields
\begin{equation}\label{eq-23}
\pi(\beta,\theta)
\;=\;
k \,\exp\!\bigl\{H(\beta,\theta)\bigr\},
\qquad \beta>0,\ \theta>0,
\end{equation}
with \(k^{-1}=\int_{0}^{\infty}\int_{0}^{\infty}\exp\{H(\beta,\theta)\}\,d\beta\,d\theta\) ensuring normalization. A straightforward evaluation gives
\[
H(\beta,\theta)
=\ln\beta
\;+\;\frac{1}{\beta}\ln\theta
\;-\;2,
\]
so that from \eqref{eq-23} one obtains
\[
\pi_M(\beta,\theta)
\;\propto\;
\exp\!\Bigl(\ln\beta + \tfrac{1}{\beta}\ln\theta -2\Bigr)
\;\propto\;
\beta\,\theta^{1/\beta}.
\]

This specification accentuates the contribution of the data via the factor \(\theta^{1/\beta}\) while imposing a mild dependence on \(\beta\). Zellner \cite{zellner1990bayesian} discusses further properties of the MDIP, including its limited invariance, which we accept given the scarcity of prior information. In our analysis, we therefore adopt
\begin{equation}\label{zellerprior}
\pi_M(\beta,\theta)\propto\beta\,\theta^{1/\beta}    
\end{equation}
as the MDIP prior for the Dhillon distribution.

The joint posterior distribution for \(\beta\) and \(\theta\), obtained using the MDIP prior \eqref{zellerprior}, is proportional to the product of the likelihood \eqref{eq-8} and the prior, yielding
\begin{equation}\label{postj-dhillon}
\begin{aligned}
p_{M}(\beta,\theta\mid\boldsymbol{t})
&\propto \beta^{n+1}\,\theta^{\,n+1/\beta}\;\prod_{i=1}^{n}
t_{i}^{\beta}\,\bigl(1+\theta\,t_{i}^{\beta}\bigr)^{-2}\,. 
\end{aligned}
\end{equation}

\begin{proposition}\label{prop2} The posterior (\ref{postj-dhillon}) is improper for all $n\in\mathbb{N}$.
\end{proposition}
\begin{proof}
    By the Fubini–Tonelli Theorem, we have
    \begin{equation}
        \begin{aligned}\label{eq-19}
            d(\boldsymbol{t}) & = \int_0^{\infty} \int_0^{\infty} \beta^{n+1} \theta^{n+1/\beta} \prod_{i=1}^{n} t_{i}^{\beta}\,\bigl(1+\theta\,t_{i}^{\beta}\bigr)^{-2} \,d\theta\,d\beta \\
            & \geq \int_0^{1/(n+1)} \int_1^{\infty} \beta^{n+1} \theta^{n+1/\beta} \prod_{i=1}^{n}t_{i}^{\beta}\,\bigl(1+\theta\,t_{i}^{\beta}\bigr)^{-2} \,d\theta\,d\beta.
        \end{aligned}
    \end{equation}
    When $\beta \leq 1/(n+1)$ and $\theta \geq 1$, we have $\theta^{n+1/\beta} \geq \theta^{2n+1}$. Thus, from (\ref{eq-19}) it follows that
    \begin{equation}\label{eq-20}
        \begin{aligned}
            d(\boldsymbol{t}) &\geq \int_0^{1/(n+1)} \beta^{n+1} \left(\prod_{i=1}^n t_i\right)^{\beta} \int_1^{\infty} \theta^{2n+1} \prod_{i=1}^{n}
            \bigl(1+\theta\,t_{i}^{\beta}\bigr)^{-2} \,d\theta\,d\beta.
        \end{aligned}
    \end{equation}
    Define
    \begin{equation*}
        \begin{aligned}
            J(\beta) = \int_{1}^{\infty} \theta^{2n+1} \prod_{i=1}^n \left(1+\theta\,t_i^\beta\right)^{-2} \,d\theta.
        \end{aligned}
    \end{equation*}
    Since
    \begin{equation*}
        \begin{aligned}
            \lim_{\theta \to \infty} \theta^{2n+1} \prod_{i=1}^n \left(1+\theta\,t_i^\beta\right)^{-2} = \infty,
        \end{aligned}
    \end{equation*}
    it follows that $J(\beta) = \infty$. As all other functions in the integrand of (\ref{eq-20})  are positive, we conclude  that $d(\boldsymbol{t}) = \infty$.
\end{proof}

Intuitively, the MDIP’s heavy-tailed behavior in the $\theta$ dimension dominates the likelihood’s decay, 
leading to divergence of the posterior normalizing constant even for large sample sizes. This is a common limitation of entropy-maximizing priors in lifetime models with scale parameters.

\subsection{Sampling from the posterior}\label{sec:sampling}

To draw inference on \(\beta\) and \(\theta\) we implement a Metropolis–Hastings sampler targeting the joint posterior in \eqref{postj-dhillon}.  Since the marginal posteriors of \(\beta\) and \(\theta\) lack closed‐form normalizing constants, we alternate between proposals based on their full conditional kernels:
\begin{align}
\pi(\theta\mid\beta,\boldsymbol{t})
&\;\propto\;
\theta^{\,n-1}\,\prod_{i=1}^n\bigl(1+\theta\,t_i^\beta\bigr)^{-2},
\label{cond:theta}\\[6pt]
\pi(\beta\mid\theta,\boldsymbol{t})
&\;\propto\;
\beta^{\,n-1}\,\prod_{i=1}^n
t_i^{\,\beta-1}\,\bigl(1+\theta\,t_i^\beta\bigr)^{-2}.
\label{cond:beta}
\end{align}

To initialize our Metropolis–Hastings sampler we employ method‐of‐moments estimates \((\beta^{(0)},\theta^{(0)})\).  Denote
$
\bar t = \frac{1}{n}\sum_{i=1}^n t_i,
$, $m_2 = \frac{1}{n}\sum_{i=1}^n t_i^2,$
and set \(R = m_2/\bar t^2\).  The initial shape \(\beta^{(0)}\) is obtained by solving
\[
\frac{m_2}{\bar t^{2}}
\;=\;
\frac{\tan (\pi/\beta)}{\pi/\beta},
\]
numerically (e.g.\ via Newton–Raphson).  Once \(\beta^{(0)}\) is found, the scale \(\theta^{(0)}\) follows from
\[
\theta^{(0)}
=\Bigl(\bar t\;\frac{\beta^{(0)}\,\sin(\pi/\beta^{(0)})}{\pi}\Bigr)^{-\beta^{(0)}}.
\]

In practical applications any sample with at least two distinct observations satisfies \(m_2>\bar t^2\), i.e.\ \(R>1\), which guarantees a unique finite solution \(\beta^{(0)}>2\).  Degenerate datasets (all \(t_i\) equal) are the only case with \(R=1\), but such samples never arise in realistic reliability studies.  Hence this initialization procedure is always well‐defined and yields valid starting values for the MCMC algorithm.

Here, we choose positive‐support proposals \(q_\theta(\theta^*\mid\theta^{(j)})\) and \(q_\beta(\beta^*\mid\beta^{(j)})\).  A convenient option is to use Gamma random walks:
\[
\theta^*\sim\mathrm{Gamma}\bigl(a_\theta,\;a_\theta/\theta^{(j)}\bigr),
\quad
\beta^*\sim\mathrm{Gamma}\bigl(a_\beta,\;a_\beta/\beta^{(j)}\bigr),
\]
with shape hyperparameters \(a_\theta,a_\beta>0\).  One may also employ log‐normal increments or any other positive‐valued kernel.
The algorithm proceeds as follows:
\begin{enumerate}
  \item Initialize \(\beta^{(1)},\theta^{(1)}\) (e.g.\ via method of moments or MLE) and set \(j=1\).
  \item \emph{Update \(\theta\):}  
    \begin{itemize}
      \item Propose \(\theta^*\sim q_\theta(\cdot\mid\theta^{(j)})\) and compute  
      \[
      A_\theta
      =\min\Biggl\{1,\;
      \frac{\pi(\theta^*\mid\beta^{(j)},\boldsymbol{t})}{\pi(\theta^{(j)}\mid\beta^{(j)},\boldsymbol{t})}\;
      \frac{q_\theta(\theta^{(j)}\mid\theta^*)}{q_\theta(\theta^*\mid\theta^{(j)})}
      \Biggr\},
      \]
      where \(\pi(\theta\mid\beta,\boldsymbol{t})\) is given by \eqref{cond:theta}.
      \item Draw \(u\sim\mathrm{Uniform}(0,1)\).  If \(u\le A_\theta\), set \(\theta^{(j+1)}=\theta^*\); otherwise \(\theta^{(j+1)}=\theta^{(j)}\).
    \end{itemize}
  \item \emph{Update \(\beta\):}  
    \begin{itemize}
      \item Propose \(\beta^*\sim q_\beta(\cdot\mid\beta^{(j)})\) and compute  
      \[
      A_\beta
      =\min\Biggl\{1,\;
      \frac{\pi(\beta^*\mid\theta^{(j+1)},\boldsymbol{t})}{\pi(\beta^{(j)}\mid\theta^{(j+1)},\boldsymbol{t})}\;
      \frac{q_\beta(\beta^{(j)}\mid\beta^*)}{q_\beta(\beta^*\mid\beta^{(j)})}
      \Biggr\},
      \]
      where \(\pi(\beta\mid\theta,\boldsymbol{t})\) is given by \eqref{cond:beta}.
      \item Draw \(u\sim\mathrm{Uniform}(0,1)\).  If \(u\le A_\beta\), set \(\beta^{(j+1)}=\beta^*\); otherwise \(\beta^{(j+1)}=\beta^{(j)}\).
    \end{itemize}
  \item Increment \(j\) and repeat steps 2–3 until convergence diagnostics are satisfied.
\end{enumerate}

Alternative proposal distributions (e.g.\ log‐normal or adaptive schemes) may be used in place of the Gamma random walks, provided they have full support on \((0,\infty)\). It is worth mentioning here that, tuning the proposal scales \(a_\theta,a_\beta\) is crucial to achieve reasonable acceptance rates (e.g.\ 20–40 \%).

\begin{table}[!b]
\centering
\caption{Performance of estimators (MM, MLE and Bayes): Bias, MSE and coverage probability (CP) for various sample sizes $n$ (with $N=10{,}000$ simulations), $\beta_1=4$ y $\theta_1=2$ .}
\begin{tabular}{c|c|cc|ccc|ccc}
\hline
Parameter & $n$ 
  & \multicolumn{2}{c|}{\textbf{MM}} 
  & \multicolumn{3}{c|}{\textbf{MLE}} 
  & \multicolumn{3}{c}{\textbf{Bayes}} \\
\hline
 & 
  & Bias & MSE  
  & Bias & MSE & CP 
  & Bias & MSE & CP \\
\hline
\multirow{11}{*}{$\beta_1$} 
&  20  & 0.726 & 4.952 & 0.284 & 3.407 & 95.05  & 0.264 & 3.351 & 93.70 \\
 &  30  & 0.536 & 4.005 & 0.171 & 2.806 & 95.45  & 0.159 & 2.775 & 94.70 \\
 &  40  & 0.437 & 3.554 & 0.118 & 2.573 & 94.80  & 0.108 & 2.550 & 93.85 \\
 &  50  & 0.388 & 3.379 & 0.115 & 2.493 & 94.25  & 0.107 & 2.478 & 93.65 \\
 &  60  & 0.352 & 3.218 & 0.091 & 2.421 & 95.25  & 0.085 & 2.407 & 94.70 \\
 &  70  & 0.306 & 3.093 & 0.072 & 2.352 & 95.40  & 0.067 & 2.340 & 95.05 \\
 &  80  & 0.295 & 2.967 & 0.072 & 2.283 & 94.95  & 0.067 & 2.273 & 94.10 \\
 &  90  & 0.273 & 2.871 & 0.064 & 2.230 & 94.80  & 0.059 & 2.221 & 94.30 \\
 & 100  & 0.227 & 2.820 & 0.062 & 2.263 & 95.60  & 0.058 & 2.257 & 95.20 \\
 & 110  & 0.239 & 2.793 & 0.050 & 2.231 & 94.65  & 0.047 & 2.225 & 94.00 \\
 & 120  & 0.211 & 2.687 & 0.043 & 2.164 & 95.30  & 0.039 & 2.158 & 94.75 \\
\hline
\multirow{11}{*}{$\theta_1$} 
 &  20  & 0.493 & 3.584 & 0.335 & 2.943 & 92.85  & 0.331 & 2.946 & 94.25 \\
 &  30  & 0.283 & 2.610 & 0.185 & 2.390 & 93.55  & 0.181 & 2.393 & 94.05 \\
 &  40  & 0.193 & 2.367 & 0.118 & 2.298 & 93.70  & 0.116 & 2.299 & 94.70 \\
 &  50  & 0.160 & 2.143 & 0.109 & 2.091 & 94.30  & 0.108 & 2.096 & 94.70 \\
 &  60  & 0.126 & 2.139 & 0.083 & 2.130 & 93.55  & 0.082 & 2.133 & 94.95 \\
 &  70  & 0.101 & 2.114 & 0.062 & 2.092 & 94.20  & 0.060 & 2.094 & 94.95 \\
 &  80  & 0.114 & 2.073 & 0.074 & 2.059 & 95.05  & 0.073 & 2.059 & 94.85 \\
 &  90  & 0.093 & 2.026 & 0.061 & 2.038 & 95.10  & 0.060 & 2.039 & 94.80 \\
 & 100  & 0.089 & 2.044 & 0.067 & 2.031 & 94.10  & 0.066 & 2.033 & 93.60 \\
 & 110  & 0.075 & 2.033 & 0.043 & 2.049 & 94.20  & 0.043 & 2.051 & 93.85 \\
 & 120  & 0.089 & 1.979 & 0.066 & 1.986 & 95.05  & 0.065 & 1.990 & 94.45 \\
\hline
\end{tabular}
\label{tabperformance}
\end{table}

\begin{table}[!t]
\centering
\caption{Performance of estimators (MM, MLE and Bayes): Bias, MSE and coverage probability (CP) for various sample sizes $n$, $\beta=3$ and $\theta=0.5$.}
\begin{tabular}{c|c|cc|ccc|ccc}
\hline
Parameter & $n$ 
  & \multicolumn{2}{c|}{\textbf{MM}} 
  & \multicolumn{3}{c|}{\textbf{MLE}} 
  & \multicolumn{3}{c}{\textbf{Bayes}} \\
\hline
 & 
  & Bias & MSE  
  & Bias & MSE & CP 
  & Bias & MSE & CP \\
\hline
\multirow{11}{*}{$\beta_2$} 
 &  20  & 0.750 & 6.093 & 0.205 & 4.075 & 95.06  & 0.190 & 4.027 & 94.28 \\
 &  30  & 0.592 & 5.336 & 0.129 & 3.704 & 95.24  & 0.119 & 3.677 & 94.49 \\
 &  40  & 0.521 & 4.991 & 0.100 & 3.549 & 95.34  & 0.093 & 3.531 & 94.81 \\
 &  50  & 0.464 & 4.786 & 0.078 & 3.481 & 95.15  & 0.072 & 3.465 & 94.66 \\
 &  60  & 0.427 & 4.604 & 0.070 & 3.418 & 94.81  & 0.065 & 3.406 & 94.34 \\
 &  70  & 0.394 & 4.487 & 0.054 & 3.358 & 95.24  & 0.050 & 3.348 & 94.60 \\
 &  80  & 0.361 & 4.359 & 0.043 & 3.319 & 95.34  & 0.039 & 3.312 & 94.77 \\
 &  90  & 0.348 & 4.299 & 0.048 & 3.322 & 95.18  & 0.045 & 3.314 & 94.84 \\
 & 100  & 0.325 & 4.226 & 0.038 & 3.290 & 95.09  & 0.035 & 3.282 & 94.69 \\
 & 110  & 0.314 & 4.166 & 0.034 & 3.275 & 95.04  & 0.032 & 3.268 & 94.58 \\
 & 120  & 0.305 & 4.126 & 0.034 & 3.260 & 94.79  & 0.031 & 3.254 & 94.60 \\
\hline
\multirow{11}{*}{$\theta_2$} 
 &  20  & -0.083 & 3.388 & 0.022 & 3.125 & 90.62  & 0.023 & 3.123 & 93.90 \\
 &  30  & -0.082 & 3.356 & 0.014 & 3.114 & 92.22  & 0.015 & 3.112 & 94.51 \\
 &  40  & -0.079 & 3.347 & 0.011 & 3.119 & 93.10  & 0.011 & 3.118 & 94.77 \\
 &  50  & -0.074 & 3.333 & 0.009 & 3.118 & 93.11  & 0.010 & 3.117 & 94.26 \\
 &  60  & -0.073 & 3.329 & 0.006 & 3.127 & 93.46  & 0.006 & 3.127 & 94.40 \\
 &  70  & -0.071 & 3.321 & 0.005 & 3.127 & 93.42  & 0.005 & 3.126 & 94.25 \\
 &  80  & -0.065 & 3.303 & 0.008 & 3.118 & 94.02  & 0.008 & 3.118 & 94.20 \\
 &  90  & -0.064 & 3.301 & 0.004 & 3.125 & 94.19  & 0.004 & 3.125 & 94.50 \\
 & 100  & -0.062 & 3.293 & 0.004 & 3.123 & 94.38  & 0.005 & 3.123 & 94.66 \\
 & 110  & -0.061 & 3.292 & 0.003 & 3.125 & 94.43  & 0.003 & 3.124 & 94.90 \\
 & 120  & -0.059 & 3.287 & 0.003 & 3.128 & 94.24  & 0.003 & 3.127 & 94.73 \\
\hline
\end{tabular}
\label{tabperformance2}
\end{table}

\section{Simulation Study}

We assess the finite‐sample behavior of Bayesian inference under different objective priors for the Dhillon distribution via a Monte Carlo experiment.  For each sample size \(n\in\{20,30,\dots,120\}\), we generate \(N=10{,}000\) independent datasets \(\{t_{1}^{(i)},\dots,t_{n}^{(i)}\}\) from \(\mathrm{Dhillon}(\beta,\theta)\) with different values for the parameters.  In order to perform Bayesian inference, we applied a Metropolis–Hastings sampler to draw from the conditional posterior distributions of \(\beta\) and \(\theta\) (see Section~\ref{sec:sampling}).  For each simulated dataset, the algorithm was run for 5 500 iterations, of which the first 500 were discarded as burn‐in.  To mitigate autocorrelation among successive draws, we retained only one sample every 5 iterations, yielding \(M = 1000\) posterior realizations per replicate.  Convergence of each chain was assessed using Geweke’s diagnostic at the 95\% confidence level, and only chains satisfying this criterion were used for inference.  From the remaining draws $\{(\beta^{(j)},\theta^{(j)})\}_{j=1}^{M},$ we computed the posterior medians \(\hat\beta_i\) and \(\hat\theta_i\) for the \(i\)th replicate, as well as 95\% credibility intervals (CIs) for uncertainty assessment.  

To quantify estimator performance, we record for each parameter \(\phi\in\{\beta,\theta\}\) the empirical bias and mean squared error (MSE),
\[
\mathrm{Bias}_\phi
=\frac{1}{N}\sum_{i=1}^N\bigl(\hat\phi_i-\phi\bigr),
\qquad
\mathrm{MSE}_\phi
=\frac{1}{N}\sum_{i=1}^N\bigl(\hat\phi_i-\phi\bigr)^2.
\]
We also evaluate the coverage frequency of nominal 95\% Bayesian credible intervals and, for comparison, of the large‐sample confidence intervals based on the inverse Fisher information.  Ideally, good estimators exhibit small Bias and MSE and coverage rates close to 0.95. All posterior computations were carried out in \textsf{R}, and the simulation code is publicly available (see Code Availability).  To simulate draws from \(\mathrm{Dhillon}(\beta,\theta)\), we apply the inverse‐transform method: if \(U\sim\mathrm{Uniform}(0,1)\), then
\begin{equation}\label{eq-quantile-dhillon}
T
=\Bigl(\tfrac{U}{\theta\,(1-U)}\Bigr)^{1/\beta}
\end{equation}
has the \(\mathrm{Dhillon}(\beta,\theta)\) law.  This closed‐form quantile function ensures efficient generation of large samples.

Due to space constraints, we focus on three representative parameter settings—Scenario 1: \((\beta,\theta)=(4,2)\), Scenario 2: \((\beta,\theta)=(3,0.5)\), and Scenario 3: \((\beta,\theta)=(8,2)\).  Tables \ref{tabperformance}-\ref{tabperformance3} report the Bias and MSE for each prior choice as the sample size \(n\) varies under these three scenarios as well as the corresponding empirical 95\% coverage probabilities.

\begin{table}[!h]
\centering
\caption{Performance of estimators (MM, MLE and Bayes): Bias, MSE and coverage probability (CP) for various sample sizes $n$ (with $N=10{,}000$ simulations), $\beta=3$ and $\theta=4$.}
\begin{tabular}{c|c|cc|ccc|ccc}
\hline
Parameter & $n$ 
  & \multicolumn{2}{c|}{\textbf{MM}} 
  & \multicolumn{3}{c|}{\textbf{MLE}} 
  & \multicolumn{3}{c}{\textbf{Bayes}} \\
\hline
 & 
  & Bias & MSE  
  & Bias & MSE & CP 
  & Bias & MSE & CP \\
\hline
\multirow{11}{*}{$\beta_3$} 
 &  20  & 0.765 & 0.871 & 0.218 & 0.758 & 94.95  & 0.203 & 0.766 & 93.81 \\
 &  30  & 0.611 & 0.647 & 0.138 & 0.635 & 95.28  & 0.128 & 0.641 & 94.26 \\
 &  40  & 0.515 & 0.557 & 0.094 & 0.590 & 95.44  & 0.087 & 0.596 & 94.79 \\
 &  50  & 0.455 & 0.511 & 0.074 & 0.567 & 95.10  & 0.068 & 0.572 & 94.66 \\
 &  60  & 0.430 & 0.479 & 0.073 & 0.546 & 94.91  & 0.068 & 0.551 & 94.26 \\
 &  70  & 0.394 & 0.484 & 0.057 & 0.543 & 95.05  & 0.053 & 0.547 & 94.20 \\
 &  80  & 0.363 & 0.457 & 0.045 & 0.541 & 95.42  & 0.041 & 0.544 & 94.77 \\
 &  90  & 0.343 & 0.454 & 0.043 & 0.529 & 95.05  & 0.040 & 0.532 & 94.57 \\
 & 100  & 0.325 & 0.440 & 0.037 & 0.524 & 95.30  & 0.034 & 0.526 & 94.56 \\
 & 110  & 0.314 & 0.446 & 0.037 & 0.521 & 95.40  & 0.035 & 0.523 & 94.94 \\
 & 120  & 0.296 & 0.441 & 0.032 & 0.522 & 95.05  & 0.030 & 0.525 & 94.76 \\
\hline
\multirow{11}{*}{$\theta_3$} 
 &  20  & 2.403 &43.194 & 1.139 & 16.346 & 92.60  & 1.117 & 16.930 & 94.07 \\
 &  30  & 1.422 &15.242 & 0.632 & 6.161 & 93.59  & 0.616 & 6.124 & 94.40 \\
 &  40  & 1.050 & 9.183 & 0.444 & 3.961 & 94.26  & 0.434 & 3.955 & 94.55 \\
 &  50  & 0.798 & 6.083 & 0.325 & 2.861 & 94.24  & 0.317 & 2.840 & 94.61 \\
 &  60  & 0.719 & 5.113 & 0.289 & 2.472 & 94.51  & 0.282 & 2.452 & 94.82 \\
 &  70  & 0.606 & 4.136 & 0.228 & 2.069 & 94.63  & 0.222 & 2.057 & 94.47 \\
 &  80  & 0.551 & 3.633 & 0.209 & 1.867 & 94.60  & 0.205 & 1.865 & 94.54 \\
 &  90  & 0.514 & 3.154 & 0.200 & 1.663 & 95.15  & 0.196 & 1.660 & 94.98 \\
 & 100  & 0.457 & 2.846 & 0.166 & 1.533 & 94.88  & 0.163 & 1.531 & 94.63 \\
 & 110  & 0.418 & 2.543 & 0.150 & 1.398 & 95.33  & 0.148 & 1.399 & 94.97 \\
 & 120  & 0.384 & 2.330 & 0.134 & 1.321 & 94.87  & 0.132 & 1.319 & 94.49 \\
\hline
\end{tabular}
\label{tabperformance3}
\end{table}

Overall, the simulation results show a consistent advantage of the Bayesian estimator under Jeffreys/reference prior 
over the classical MLE, especially in small to moderate samples, with noticeable reductions in bias and MSE. 
The coverage probabilities for Bayesian credibility intervals remain close to the nominal 95\% across all scenarios, 
while the frequentist confidence intervals tend to undercover for the scale parameter $\theta$ in small samples. 
These findings confirm the robustness of the proposed objective Bayesian approach and its practical advantage 
in reliability contexts where data scarcity is common.

\ \

\section{Posterior Predictive Distribution}

To forecast a future lifetime \(T_{\rm new}\) under the Dhillon model, we start from the exact Bayesian identity
\[
p\bigl(T_{\rm new}\mid \boldsymbol{t}\bigr)
\;=\;
\iint
f\bigl(T_{\rm new};\beta,\theta\bigr)\;
p\bigl(\beta,\theta\mid \boldsymbol{t}\bigr)
\;d\beta\,d\theta,
\]
where \(f(\cdot;\beta,\theta)\) is the Dhillon density \eqref{eq-1} and \(p(\beta,\theta\mid\boldsymbol{t})\) is the joint posterior \eqref{postr-dhillon}.  Since this double integral has no closed‐form solution, we approximate it by simulation:

1. Obtain Monte Carlo samples \(\{(\beta^{(j)},\theta^{(j)})\}_{j=1}^N\) from \(p(\beta,\theta\mid\boldsymbol{t})\) via the algorithm of Section \ref{sec:sampling}.

2. For each draw \((\beta^{(j)},\theta^{(j)})\), generate a predictive replicate
\[
T_{\rm new}^{(j)}\;\sim\;\mathrm{Dhillon}\bigl(\beta^{(j)},\theta^{(j)}\bigr).
\]

3. Treat \(\{T_{\rm new}^{(j)}\}_{j=1}^N\) as an empirical sample from the posterior predictive distribution.

From these simulated values we may compute, for example, the predictive mean,
\[
\widehat{T}_{\rm pred}
=\frac{1}{N}\sum_{j=1}^N T_{\rm new}^{(j)},
\]
and a \((1-\gamma)\times100\%\) credible interval by taking the \(\gamma/2\) and \(1-\gamma/2\) quantiles of \(\{T_{\rm new}^{(j)}\}\).

Thus, this Monte Carlo approach naturally propagates uncertainty in \(\beta\) and \(\theta\) into predictions for future lifetimes under the Dhillon distribution.

\section{Reliability of Harvester Components}

In contemporary sugarcane harvesting operations, unplanned stoppages of the self-propelled harvester lead to significant productivity losses and sharply higher maintenance expenses. Over a typical season, each machine processes more than 20 tons of cane per hour under variable loads and highly abrasive field conditions, which accelerate the wear and fatigue of its mechanical subsystems. To illustrate the applicability of the three-parameter Dhillon distribution, we focus on two representative data sets drawn from field records: the diesel engine and the line divider.

The diesel engine dataset comprises 66 recorded times (in days) between successive corrective interventions, reflecting the engine’s variable torque demands and thermal stresses under load. The line divider dataset contains 90 time-to-failure observations, capturing fatigue of the feed-separator mechanism as it directs cane stalks into the processing units. Modeling these failure times accurately is critical for a predictive-maintenance strategy that optimizes uptime, limits repair costs, and extends component service life.

Table \ref{table:el_failure_times} lists the observed days until failure for both the Diesel Engine and the Line Divider. In the sequel, we fit and compare five three-parameter lifetime distributions—Weibull, Gamma, Exponential–Logarithmic (EL), Weighted Lindley (WL), and Exponential–Poisson (EP)—to each of these datasets. We then assess model fit using likelihood-based criteria and examine the Dhillon distribution’s flexibility in capturing the empirical hazard shapes implied by these two components.

\begin{table}[!h]
  \centering
  \caption{Failure times for Diesel Engine and Line Divider equipment (hours)}
  \begin{tabularx}{\textwidth}{@{}lX@{}}
    \toprule
    \textbf{Component} & \multicolumn{1}{c}{\textbf{Times to Failure}} \\
    \midrule
    Diesel Engine & 1, 1, 1, 1, 1, 1, 1, 1, 1, 1, 1, 1, 1, 1, 1, 1, 1, 2, 2, 2, 2, 2, 2, 2, 2, 3, 3, 3, 3, 3, 4, 4, 4, 4, 5, 5, 6, 7, 7, 7, 7, 8, 9, 11, 11, 11, 13, 14, 15, 15, 16, 18, 21, 21, 21, 22, 25, 26, 28, 32, 52, 59 \\
    Line Divider & 1, 1, 1, 1, 1, 1, 1, 1, 1, 1, 1, 1, 1, 1, 1, 1, 1, 1, 1, 1, 1, 1, 1, 2, 2, 2, 2, 2, 2, 2, 2, 3, 3, 3, 3, 3, 3, 4, 4, 4, 4, 4, 4, 5, 5, 5, 5, 5, 6, 6, 6, 6, 6, 6, 7, 7, 8, 8, 8, 8, 9, 11, 11, 11, 11, 11, 11, 11, 12, 14, 14, 15, 17, 17, 18, 19, 21, 24, 29, 31, 32, 34 \\
    \bottomrule
  \end{tabularx}
  \label{table:el_failure_times}
\end{table}

To evaluate model fit, we use the Bayesian Information Criterion (BIC), the Akaike Information Criterion (AIC), and its small-sample correction (AICc):
\[
\mathrm{BIC} = -2\ell(\hat{\theta};t) + k\log n,\quad
\mathrm{AIC} = -2\ell(\hat{\theta};t) + 2k,\quad
\mathrm{AICc} = \mathrm{AIC} + \frac{2k(k+1)}{n - k - 1},
\]
where \(n\) is the sample size, \(k=2\) the number of parameters per model, and \(\ell(\hat{\theta};t)\) the maximized log‑likelihood. The model with the lowest AIC is selected for its balance of fit quality and simplicity.

Table \ref{criteria_by_component} shows that the Dhillon model attains the lowest AIC and AICc for both Diesel Engine and Line Divider, outperforming the other fitted distributions. Figure \ref{graf-pic-a} confirms the quality of the fit: the Dhillon parametric survival curves lie very close to the empirical Kaplan–Meier steps, indicating that the model captures the observed failure-time behavior with minimal discrepancy. Together, these results support that Dhillon provides the best balance of fit and fidelity to the data among the candidates.

\begin{table}[!h]
  \centering
  \caption{BIC, AIC, and AICc values by fitted distribution and component}
  \renewcommand{\arraystretch}{1.3}   % increase row height
  \setlength{\tabcolsep}{12pt}         % increase column separation
  \small
  \begin{tabular}{l  rrr  rrr}
    \hline
    Model    & \multicolumn{3}{c}{Diesel Engine} & \multicolumn{3}{c}{Line Divider} \\
             & BIC    & AIC    & AICc    & BIC             & AIC        & AICc          \\
    \hline
    Dhillon  & 396.11 & 391.86 & 392.06 & 488.32          & 483.51     & 483.66        \\
    Weibull  & 400.49 & 396.24 & 396.44 & 490.90          & 486.08     & 486.23        \\
    Gamma    & 402.31 & 398.06 & 398.26 & 491.08          & 486.26     & 486.41        \\
    EEG      & 396.34 & 392.09 & 392.29 & 489.25          & 484.43     & 484.59        \\
    WL       & 415.01 & 410.75 & 410.96 & 500.70          & 495.88     & 496.03        \\
    GE       & 402.59 & 398.33 & 398.54 & 491.05          & 486.23     & 486.38        \\
    EP       & 404.29 & 400.02 & 400.22 & 491.08          & 486.27     & 486.42        \\
    \hline
  \end{tabular}
  \label{criteria_by_component}
\end{table}

\begin{figure}[!htb]
  \centering
\includegraphics[scale=0.67]{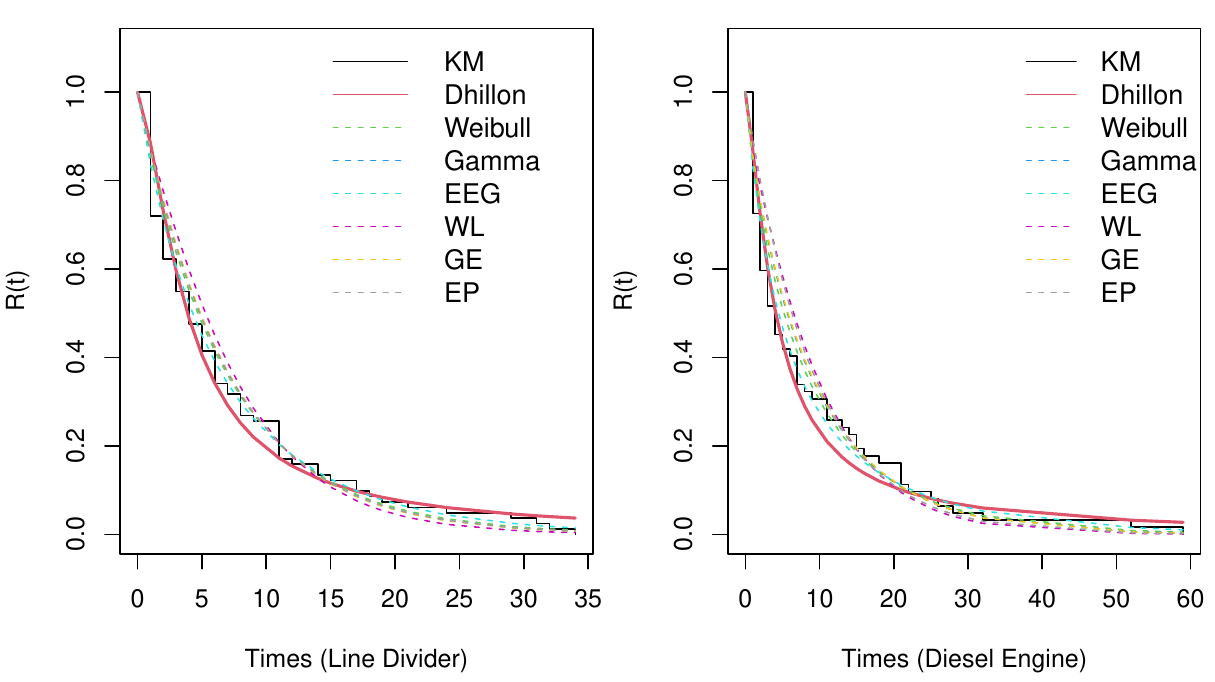}	
  \caption{Parametric survival curves overlaid on the empirical estimate for each component}
  \label{graf-pic-a}
\end{figure}

Table \ref{tab:Dhillon_sidebyside} presents the Bayesian posterior means, standard deviations, and 95\% credibility intervals for \(\beta\), \(\theta\) and the predictive \(y^*\) for both Diesel Engine and Line Divider.
\begin{table}[!h]
  \caption{Bayes estimator, standard deviation and 95\% credibility intervals for \(\mu\) (=\(\beta\)), \(\Omega\) (=\(\theta\)) and \(y^*\), comparando Diesel Engine e Line Divider}
  \centering
  {\small
  \begin{tabular}{crrrrrrrrr}
    \hline
    Parameter & \multicolumn{4}{c}{Diesel Engine} & & \multicolumn{4}{c}{Line Divider} \\
    \cline{2-5} \cline{7-10}
              & Bayes & SD & CI$_{2.5\%}$ & CI$_{97.5\%}$ & & Bayes & SD & CI$_{2.5\%}$ & CI$_{97.5\%}$ \\
    \hline
    \(\beta\) & 1.35 & 0.14 & 1.09 & 1.64 & & 1.52 & 0.14 & 1.25 & 1.79 \\
    \(\theta\) & 0.16 & 0.05 & 0.08 & 0.26 & & 0.13 & 0.04 & 0.07 & 0.21 \\
    \(T_{\rm pred}^*\) & 12.81 & 17.59 & 0.27 & 69.20 & & 9.54 & 12.04 & 0.32 & 47.53 \\
    \hline
  \end{tabular}
  \label{tab:Dhillon_sidebyside}
  }
\end{table}

The posterior predictive summaries \(T_{\rm pred}^*\)  quantify the expected future failure times under the Dhillon model, incorporating parameter uncertainty. The Diesel Engine shows a larger posterior mean predictive failure time (12.81) than the Line Divider (9.54), suggesting a tendency toward later failures for the former. However, both components exhibit very wide 95\% credibility intervals, indicating substantial uncertainty in individual predictions. This dispersion implies that, despite a central tendency toward moderate-to-long lifetimes, there is still a non-negligible probability of early failure. In practical terms, such uncertainty recommends a maintenance policy that combines routine monitoring with adaptive triggers, so that rare but early failures are not overlooked even if average predictions appear favorable.

\section{Conclusions}\label{sec:conclusions}

This paper develops a fully objective Bayesian treatment for the two-parameter Dhillon distribution, closing a long-standing gap between its practical appeal in reliability work and the absence of a rigorous Bayesian formulation. We derive the Fisher information, obtain Jeffreys' prior in closed form, and show that the ordered reference prior coincides with Jeffreys' prior, both taking the simple product form. We prove posterior propriety under mild, verifiable conditions (in particular, $n\ge2$ with at least two distinct observations), and we provide sufficient conditions for the existence of posterior moments, extending previous results on posterior propriety in complex survival models \cite{ramos2020posterior, ramos2023power, ramos2025posterior, achire2025posterior}. Importantly, we establish that Zellner's maximal data information prior \cite{zellner1, zellner2} produces an \emph{improper} posterior for all sample sizes in this model because of its heavy-tailed behavior in the scale dimension, a cautionary result for entropy-based priors in lifetime analysis. 

Computationally, we propose a simple Metropolis--Hastings scheme with positive-support proposals and a moment-based initialization that is easy to implement and empirically stable. This makes posterior sampling routine even in moderate samples and enables seamless posterior predictive calculations; the closed-form quantile function of the Dhillon model further simplifies predictive simulation and uncertainty quantification for future lifetimes.

A comprehensive Monte Carlo study shows that the objective Bayesian estimators under Jeffreys/reference prior dominate classical MLEs in small and moderate samples: biases and mean squared errors are consistently reduced, and nominal $95\%$ credibility is well tracked across scenarios featuring both decreasing and unimodal hazards. By contrast, large-sample normal intervals based on the inverse Fisher information tend to undercover the scale parameter in small $n$, an effect the Bayesian intervals mitigate \cite{berger1992reference, bernardo2005}.

Two real datasets from sugarcane-harvester components illustrate the practical gains. Against several competing parametric models, the Dhillon distribution achieves the best BIC/AIC in both cases and closely tracks the Kaplan--Meier curves, indicating superior fidelity to the observed failure mechanisms. Posterior predictive summaries reveal heavy right tails together with non-negligible early-failure mass—actionable information for condition-based maintenance: they support setting just-in-time intervention thresholds that hedge against early breakdowns without sacrificing uptime.

Methodologically, the proposed framework can be expanded to address challenges that frequently arise in applied reliability analysis. Incorporating censoring and truncation mechanisms, modelling covariate effects through proportional hazards or accelerated failure time structures \cite{achcar1992}, and accommodating long-term survivor fractions or competing risks \cite{perdona2011} would extend its reach to a wider range of field situations. In addition, hierarchical formulations could enable information sharing across related systems or operational sites, and robustness could be enhanced through heavy-tailed likelihoods or tempered posteriors. On the theoretical side, developing objective priors that retain posterior propriety under these more complex settings would strengthen the Bayesian foundations of the model. By combining these methodological advances with the inferential rigor established here, the Dhillon distribution can serve as a versatile and reliable tool for lifecycle modelling and predictive maintenance in demanding industrial environments \cite{he2023bayesian, vidovic2024properties, dong2024bayesian, shakhatreh2021objective}.

\section*{Code Availability}

The Metropolis-Hastings algorithm for sampling from the posterior distribution has been implemented in a package that is easy to use and available at: \\ \url{https://github.com/ramospedroluiz/dhillon}

\section*{Disclosure statement}

No potential conflict of interest was reported by the author(s)

\section*{Acknowledgements}

The authors are thankful to the Editorial Board and to the reviewers for their valuable comments and suggestions which led to this improved version. 
%The research was partially supported by CNPq, FAPESP and CAPES of Brazil.

\bibliographystyle{tfs}

\appendix

\section{Derivation of Moments, Mean Residual Life function, and Fisher Information Matrix for the Dhillon Distribution}\label{apendix-A}

The $r$-th moment of a random variable $T \sim \mathrm{Dhillon}(\beta,\theta)$ is given by
\begin{equation}
    \begin{aligned}
        E(T^r) & = \int_0^\infty t^r f(t;\beta,\theta) \, dt  \\
        & = \int_0^{\infty} t^r \frac{\beta\theta \, t^{\beta-1}}{(1+\theta t^\beta)^2} \, dt 
        \quad \text{(substituting $f$ from equation (\ref{eq-0}))} \\
        & = \theta^{-r/\beta}  \int_0^{\infty} \frac{x^{r/\beta}}{(1+x)^2} \, dx 
        \quad  \text{(using the change of variable $x=\theta t^\beta$)} \\
        & = \theta^{-r/\beta}  \int_0^{1} u^{r/\beta} (1-u)^{-r/\beta} \, du 
        \quad  \text{(using the change of variable $u=\frac{x}{x+1}$)} \\
        & = \theta^{-r/\beta} B(1+r/\beta,1-r/\beta) 
        \quad \text{(by the definition of the Beta function)} \\
        & = \theta^{-r/\beta} \frac{\pi r/\beta}{\sin(\pi r /\beta)} 
        \quad \text{(by the reflection identity for the Beta function).}
    \end{aligned}
\end{equation}

The mean residual life function of $T \sim \mathrm{Dhillon}(\beta,\theta)$ is given by
\begin{equation}
    \begin{aligned}
        r(t) & = \frac{1}{R(t)} \int_t^{\infty} R(s) \, ds \\
        & = (1+\theta t^\beta)  \int_t^{\infty} \frac{1}{1+\theta s^\beta} \, ds 
        \quad \text{(substituting $R$ from equation (\ref{survival}))} \\
        & = \frac{\theta^{-1/\beta}(1+\theta t^\beta)}{\beta} 
        \int_{\theta t^\beta}^\infty x^{1/\beta-1} (1+x)^{-1} \, dx 
        \quad \text{(using the change of variable $x=\theta s^\beta$)} \\
        & = \frac{\theta^{-1/\beta}(1+\theta t^\beta)}{\beta} 
        \int_0^{\frac{1}{1+\theta t^\beta}} u^{-1/\beta} (1-u)^{1/\beta-1} \, du 
        \quad \text{(using the change of variable $u=\frac{1}{x+1}$)} \\
        & = \frac{\theta^{-1/\beta}(1+\theta t^\beta)}{\beta} 
        B\left(\frac{1}{1+\theta t^\beta}; 1-\frac{1}{\beta},\frac{1}{\beta}\right) 
        \quad \text{(where $B$ is the incomplete Beta function).}
    \end{aligned}
\end{equation}

The log-likelihood function for a single observation $t$ from $T \sim \mathrm{Dhillon}(\beta,\theta)$ is
\begin{equation}
    l(t;\beta,\theta) = \log \beta + \log \theta + (\beta-1)\log t - 2\log(1+\theta t^\beta).
\end{equation}

The first and second derivatives of $l(t;\beta,\theta)$ are
\begin{equation}
\begin{aligned}
    \frac{\partial l}{\partial \beta} &= \frac{1}{\beta}+\log t-\frac{2\theta t^\beta \log t}{1+\theta t^\beta}, \\
    \frac{\partial l}{\partial \theta} &= \frac{1}{\theta}-\frac{2 t^\beta}{1+\theta t^\beta},
\end{aligned} 
\end{equation}
and
\begin{equation}
\begin{aligned}
    \frac{\partial^2 l}{\partial \beta^2} &= -\frac{1}{\beta^2}-\frac{2\theta t^\beta \log^2 t}{(1+\theta t^\beta)^2}, \\
    \frac{\partial^2 l}{\partial \theta \partial \beta} &= -\frac{2 t^\beta \log t}{(1+\theta t^\beta)^2}, \\
    \frac{\partial^2 l}{\partial \theta^2} &= -\frac{1}{\theta^2}+\frac{2 t^{2\beta}}{(1+\theta t^\beta)^2}.
\end{aligned} 
\end{equation}

By straightforward integration, we obtain:
\begin{equation}
    I_{\beta\beta} = \frac{1}{\beta^2} + 2 E\left[ \frac{\theta T^\beta \log^2 T}{(1+\theta T^\beta)^2} \right]  
      = \frac{1}{9\beta^2} \left[ \pi^2 + 3 + 3\log^2\theta \right],
\end{equation}
\begin{equation}
    I_{\beta\theta} = 2 E\left[ \frac{T^\beta \log T}{(1+\theta T^\beta)^2} \right] 
          =  -\frac{\log\theta}{3\theta\beta},
\end{equation}
and
\begin{equation}
    I_{\theta\theta} = \frac{1}{\theta^2} - 2 E\left[ \frac{T^{2\beta}}{(1+\theta T^\beta)^2} \right] 
          = \frac{1}{3\theta^2}.
\end{equation}

\section{Existence of Posterior Moments}
\begin{proposition}\label{prop-moments-beta}
The posterior moments of $\beta$, under the posterior distribution (\ref{postj-dhillon}), are finite for $n \geq 2$, provided that not all $t_i$ are equal.
\end{proposition}
\begin{proof}
We aim to show the convergence of the following integral:
\[
E(\beta^{m} \mid \boldsymbol{t}) \propto \int_0^\infty \int_0^\infty
\beta^{n+m-1} \, \theta^{-1} \prod_{i=1}^n \theta t_i^{\beta} (1 + \theta t_i^\beta)^{-2} \, d\theta \, d\beta.
\]

Proceeding as in the proof of the proposition \ref{prop2}, we obtain the bound
\begin{equation}\label{moments-of-beta}
    E(\beta^{m} \mid \boldsymbol{t}) \leq \int_0^\infty \beta^{n+m-1} r^\beta J(\beta) \, d\beta,
\end{equation}
where \(J(\beta)\) is defined in equation~(\ref{J-integral}). Since \(J(\beta)\) grows at most polynomially, the integral on the right side of inequality~\eqref{moments-of-beta} converges. Therefore, we conclude that \(E(\beta^{m} \mid \boldsymbol{t})<\infty\).
\end{proof}

\begin{proposition}\label{prop-moments-theta}
The \(m\)-th posterior moment of \(\theta\) under the posterior distribution (\ref{postj-dhillon}) is finite 
provided that the sample satisfies the following conditions:
\begin{enumerate}
    \item[(I)] there is at least one observation less than \(1\),
    \item[(II)] there are at least \(m+1\) observations greater than \(1\).
\end{enumerate}
\end{proposition}

\begin{proof}
Let \(t_1 < t_2 < \dots < t_n\) be observations from the \(Dhillon(\beta,\theta)\) distribution. Our goal is to show that the posterior moment given below is convergent:
\begin{equation*}
    \begin{aligned}
    E(\theta^m \mid \boldsymbol{t}) \propto \int_0^\infty  \int_0^\infty \beta^{n-1}\theta^{n+m-1} \prod_{i=1}^n t_i^{\beta-1}\left(1 + \theta t_i^\beta\right)^{-2}  \, d\beta \, d\theta.
    \end{aligned}
\end{equation*}
We decompose $E(\theta^m \mid \boldsymbol{t})\propto s_1+s_2+s_3+s_4$ where
    \begin{equation*}
        \begin{aligned}          
s_1 & =\int_{0}^{1}\!\!\int_{0}^{1}
\beta^{n-1}\,\theta^{n+m-1}
    \;\displaystyle\prod_{i=1}^n
\,t_i^{\beta-1}\bigl(1+\theta\,t_i^\beta\bigr)^{-2}
\;d\beta\,d\theta \\
s_2 & =\int_{1}^{\infty}\!\!\int_{0}^{1}
\beta^{n-1}\,\theta^{n+m-1}
    \;\displaystyle\prod_{i=1}^n
\,t_i^{\beta-1}\bigl(1+\theta\,t_i^\beta\bigr)^{-2}
\;d\beta\,d\theta \\
s_3 & =\int_{0}^{1}\!\!\int_{1}^{\infty}
\beta^{n-1}\,\theta^{n+m-1}
    \;\displaystyle\prod_{i=1}^n
\,t_i^{\beta-1}\bigl(1+\theta\,t_i^\beta\bigr)^{-2}
\;d\beta\,d\theta \\
s_4 & =\int_{1}^{\infty}\!\!\int_{1}^{\infty}
\beta^{n-1}\,\theta^{n+m-1}
    \;\displaystyle\prod_{i=1}^n
\,t_i^{\beta-1}\bigl(1+\theta\,t_i^\beta\bigr)^{-2}
\;d\beta\,d\theta \\
        \end{aligned}
    \end{equation*}

For $s_1$ we use that $(1+\theta t_i^{\beta})^{-2}<1$, thus
\begin{equation}\label{s1}
    \begin{aligned}
        s_1 & \leq \int_{0}^{1}\!\!\int_{0}^{1}
\beta^{n-1}\,\theta^{n+m-1}
    \;\displaystyle\prod_{i=1}^n
\,t_i^{\beta-1}
\;d\beta\,d\theta \\
 & = \int_{0}^{1}\!\!\int_{0}^{1}
\beta^{n-1}\,\theta^{n+m-1}
    \;\left(\prod_{i=1}^n
\,t_i\right)^{\beta-1}
\;d\beta\,d\theta \\
& = \int_{0}^{1}
 \beta^{n-1} \;\left(\prod_{i=1}^n
\,t_i\right)^{\beta-1}
\;d\beta \int_{0}^{1} \theta^{n+m-1}
    \,d\theta \\
 & \propto \int_{0}^{1}
  \beta^{n-1}
\;d\beta \int_{0}^{1} \theta^{n+m-1}
    \,d\theta \\
    \end{aligned}
\end{equation}
The univariate integrals in (\ref{s1}) converge for all $n>0$ and $m>0$, therefore $s_1<\infty$.

For $s_2$ we use that $(1+\theta t_i^{\beta})^{-2}<\theta^{-2} t_i^{-2\beta}$, thus
\begin{equation}\label{s2}
    \begin{aligned}
        s_2 & \leq \int_{1}^{\infty}\!\!\int_{0}^{1}
\beta^{n-1}\,\theta^{n+m-1}
    \;\displaystyle\prod_{i=1}^n
\, \theta^{-2} t_i^{-\beta-1}
\;d\beta\,d\theta \\
 & = \int_{1}^{\infty}\!\!\int_{0}^{1}
\beta^{n-1}\,\theta^{-n+m-1}
    \;\left(\prod_{i=1}^n
\,t_i\right)^{-\beta-1}
\;d\beta\,d\theta \\
& = \int_{0}^{1}
 \beta^{n-1} \;\left(\prod_{i=1}^n
\,t_i\right)^{-\beta-1}
\;d\beta \int_{1}^{\infty} \theta^{-n+m-1}
    \,d\theta \\
 & \propto \int_{0}^{1}
  \beta^{n-1}
\;d\beta \int_{1}^{\infty} \theta^{-n+m-1}
    \,d\theta \\
    \end{aligned}
\end{equation}
By hypothesis (II), $m<n$; hence, the univariate integrals in (\ref{s2}) converge, implying that $s_2<\infty$.

For $s_3$ we use that $(1+\theta t_1^{\beta})^{-2}<1$ and that  $\theta t_i^\beta (1+\theta t_i^{\beta})^{-2}<1$, thus
\begin{equation}\label{s3}
    \begin{aligned}
        s_3 &  =  \int_{0}^{1}\!\!\int_{1}^{\infty}
\beta^{n-1}\,\theta^{m-1}
    \;\displaystyle\prod_{i=1}^1
\,\theta t_i^{\beta-1}\bigl(1+\theta\,t_i^\beta\bigr)^{-2}\; \displaystyle\prod_{i=2}^n
\,\theta t_i^{\beta-1}\bigl(1+\theta\,t_i^\beta\bigr)^{-2}
\;d\beta\,d\theta \\
& \propto \int_{0}^{1}\!\!\int_{1}^{\infty}
\beta^{n-1}\,\theta^{m-1}
    \;\displaystyle\prod_{i=1}^1
\,\theta t_i^{\beta}\bigl(1+\theta\,t_i^\beta\bigr)^{-2}\; \displaystyle\prod_{i=2}^n
\,\theta t_i^{\beta}\bigl(1+\theta\,t_i^\beta\bigr)^{-2}
\;d\beta\,d\theta \\
& \leq \int_{0}^{1}\!\!\int_{1}^{\infty}
\beta^{n-1}\,\theta^{m}
\, t_1^{\beta}
\;d\beta\,d\theta \\
 & \lesssim \int_{1}^{\infty}
  \beta^{n-1} t_1^\beta
\,d\beta \int_{0}^{1} \theta^{m}
    \,d\theta \\
    \end{aligned}
\end{equation}
The last integral in (\ref{s3}) converges since, by hypothesis (I), we have \(t_1 < 1\), which ensures that \(s_3 < \infty\).

For $s_4$, we decompose the integrand into two factors.  
When $t_i>1$, we apply the bound $(1+\theta t_i^{\beta})^{-2} < \theta^{-2} t_i^{-2\beta}$,  
whereas for $t_i<1$, we use $\theta t_i^\beta (1+\theta t_i^{\beta})^{-2} < 1$.  
Therefore,

\begin{equation}\label{s4}
    \begin{aligned}
    s_4 &  =  \int_{1}^{\infty}\!\!\int_{1}^{\infty}
\beta^{n-1}\,\theta^{m-1}
    \;\displaystyle\prod_{t_i<1}
\,\theta t_i^{\beta-1}\bigl(1+\theta\,t_i^\beta\bigr)^{-2}\; \displaystyle\prod_{t_i>1}
\,\theta t_i^{\beta-1}\bigl(1+\theta\,t_i^\beta\bigr)^{-2}
\;d\beta\,d\theta \\
& \propto \int_{1}^{\infty}\!\!\int_{1}^{\infty}
\beta^{n-1}\,\theta^{m-1}
    \;\displaystyle\prod_{t_i<1}
\,\theta t_i^{\beta}\bigl(1+\theta\,t_i^\beta\bigr)^{-2}\; \displaystyle\prod_{t_i>1}
\,\theta t_i^{\beta}\bigl(1+\theta\,t_i^\beta\bigr)^{-2}
\;d\beta\,d\theta \\
    & \leq \int_{1}^{\infty}\!\!\int_{1}^{\infty}
\beta^{n-1}\,\theta^{m-1}
    \; \displaystyle\prod_{t_i>1}
\,\theta^{-1} t_i^{-\beta}
\;d\beta\,d\theta \\
 & = \int_{1}^{\infty}\!\!\int_{1}^{\infty}
\beta^{n-1}\,\theta^{-k+m-1}
    \;\left(\prod_{t_i>1}
\,t_i\right)^{-\beta}
\;d\beta\,d\theta \\
& = \int_{1}^{\infty}
\beta^{n-1} \;\left(\prod_{t_i>1}
\,t_i\right)^{-\beta}
\;d\beta \int_{1}^{\infty} \theta^{-k+m-1}
    \,d\theta, \\
    \end{aligned}
\end{equation}
where $k$ denotes the number of observations strictly less than $1$. 

By hypothesis (II), we have \(k > m\), which ensures that the integral with respect to \(\theta\) in (\ref{s4}) converges. 
Moreover, since the product \(\prod_{t_i > 1} t_i > 1\), the first univariate integral also converges. 
Therefore, we conclude that \(s_4 < \infty\).

Under conditions (I) and (II), we have shown that \(s_i < \infty\) for \(i=1,2,3,4\), and thus, we conclude that 
\[
E(\theta^m \mid \boldsymbol{t}) < \infty.
\]

\end{proof}
\begin{remark}\label{rmk-1}
The posterior moments for \(\theta\) may not always exist. For example, consider the case \(n=2\) with \(t_1 < t_2\). The posterior mean of \(\theta\) is given by
\begin{equation}
    \begin{aligned}
    E(\theta \mid \boldsymbol{t}) = \int_0^\infty \beta^{n-1} \int_0^\infty \prod_{i=1}^2 \frac{\theta t_i^\beta}{(1 + \theta t_i^\beta)^2} \, d\theta \, d\beta.
    \end{aligned}
\end{equation}

By applying the change of variable \(x = 1/\theta\) and setting \(a_i = t_i^\beta\), the integral can be rewritten as
\begin{equation}\label{eq-14}
    \begin{aligned}
    E(\theta \mid \boldsymbol{t}) 
    &= \int_0^\infty \beta^{n-1} \int_0^\infty \frac{a_1 a_2}{(x + a_1)^2 (x + a_2)^2} \, dx \, d\beta \\
    &= \int_0^\infty \beta^{n-1} \left( \frac{a_1^2 - a_2^2 + 2 a_1 a_2 \ln\left(\frac{a_2}{a_1}\right)}{(a_1 - a_2)^3} \right) d\beta \\
    &= \int_0^\infty \beta^{n-1} \left( \frac{\left(\frac{a_1}{a_2}\right)^2 - 1 - 2 \frac{a_1}{a_2} \ln\left(\frac{a_1}{a_2}\right)}{a_2 \left(\frac{a_1}{a_2} - 1\right)^3} \right) d\beta \\
    &= \int_0^\infty \beta^{n-1} t_2^{-\beta} \left( \frac{r^{2\beta} - 1 - 2 r^\beta \beta \ln r}{(r^\beta - 1)^3} \right) d\beta,
    \end{aligned}
\end{equation}
where \(r = \frac{t_1}{t_2} < 1\).

The ratio inside the integrand of \eqref{eq-14} is bounded. Therefore,
\[
E(\theta \mid \boldsymbol{t}) \propto \int_0^\infty \beta^{n-1} t_2^{-\beta} \, d\beta.
\]
If \(t_2 < 1\), the integral diverges and hence the posterior mean \(E(\theta \mid \boldsymbol{t})\) is infinite.

 \end{remark}

\end{document}